\documentclass[a4paper,10pt]{article}
\usepackage{graphicx} 
\usepackage[utf8]{inputenc}
\usepackage[T1]{fontenc}
\usepackage{comment}
\usepackage{float}
\usepackage{natbib}
\usepackage{pdflscape}  
\usepackage{xcolor}
\usepackage{booktabs} 
\usepackage{lipsum} 

\usepackage{a4wide}
\usepackage{indentfirst}
\usepackage{authblk}

\usepackage{amsmath,amssymb}
\usepackage{amsthm}
\usepackage[colorlinks=true, linkcolor=black, citecolor=black, urlcolor=blue]{hyperref}
\usepackage{footmisc} 
\usepackage{subfig}
\usepackage{graphicx}
\usepackage{algorithm2e}
\usepackage{subcaption}
\usepackage{chngcntr}
\counterwithin{table}{section}
\counterwithin{figure}{section}

\usepackage[margin=1.2in]{geometry}
\usepackage{graphicx}
\usepackage{listings}
\usepackage[utf8]{inputenc}
\usepackage{appendix}
\usepackage{shuffle}
\numberwithin{equation}{section}
\newtheorem{teo}{Theorem}[section]
\newtheorem*{teo*}{Theorem}
\newtheorem*{prop*}{Proposition}
\newtheorem*{corol*}{Corollary}
\newtheorem{prop}[teo]{Proposition}

\newtheorem{lema}[teo]{Lemma}
\newtheorem{defi}[teo]{Definition}

\theoremstyle{definition}

\newtheorem*{remark*}{Remark}
\newtheorem{ex}[teo]{Example}

\newtheorem{remark}[teo]{Remark}

\newcommand{\N}{\mathbb{N}}

\newcommand{\R}{\mathbb{R}}

\newcommand{\blfootnote}[2]{%
  \begingroup
  \renewcommand\thefootnote{#1}%
  \addtocounter{footnote}{1}%
  \footnotetext{#2}%
  \addtocounter{footnote}{-1}%
  \endgroup
}

\usepackage[most]{tcolorbox}

\newtcolorbox{methodbox}[1]{
  colback=gray!5, colframe=black, boxrule=0.6pt,
  title={#1}, fonttitle=\bfseries,
  before skip=10pt, after skip=10pt
}

\title{Matrix Approximation of Bachelier Option Prices and Greeks under Stochastic Volatility models}
\author{Elisa Alòs\textsuperscript{*}, Òscar Burés\textsuperscript{†}}
\date{\today}

\begin{document}

\maketitle

\blfootnote{*}{Department of Economics and Business, Universitat Pompeu Fabra and Barcelona School of Economics. Ramón Trias Fargas 25-27, 08005, Barcelona, Spain.}
\blfootnote{†}{Departament de Matemàtica Econòmica, Financera i Actuarial, Universitat de Barcelona. Diagonal 690--696, 08034 Barcelona, Spain.}

\begingroup
\renewcommand{\thefootnote}{}  
\renewcommand{\footnotemargin}{0pt}  
\footnotetext{%
 \noindent\hspace{0pt}Òscar Burés supported by program AGAUR-FI ajuts (2025 FI-1 00580) from the Department of Research and Universities of the Government of Catalonia and the co-funding of the European Social Fund Plus (ESF+).}
\endgroup
\begin{abstract}
In this paper, we present a numerical method for option pricing and the computation of Greeks under stochastic volatility Bachelier-type models, based on elementary linear algebra. The method allows option prices and Greeks to be computed for infinitely many strikes (within a range of convergence) by evaluating only a finite number of expectations, independent of the number of strikes. For the SABR model, we derive an explicit range of convergence. Numerical examples are provided for both the SABR and the rough Bergomi models.
\end{abstract}

\noindent\textbf{Keywords:} Bachelier-type model, SABR, rough volatilities, option pricing, numerical method.

\medskip

\noindent \textbf{MSC Classification:} 91G20, 91G60, 60H30, 60G22
\section{Introduction}

The Black--Scholes model remains the cornerstone of option pricing theory, and the vast majority of models used in practice, such as local volatility, stochastic volatility, and their variants—can be viewed as refinements built on top of it. A defining feature of this framework is that asset prices are restricted to be positive, an assumption that does not always hold: negative prices have recently been recorded in commodity markets, for instance. This has led some markets to adopt the Bachelier model instead (see \cite{ASENS_1900_3_17__21_0,choi2022black}), which dispenses with the positivity constraint and treats asset prices as normally distributed.

A central challenge in option pricing under either the Black--Scholes or the Bachelier framework is finding tractable closed-form approximations for option prices and their implied volatilities. A substantial body of literature addresses this by building expansions around a leading-order term, typically the Black--Scholes or Bachelier price evaluated at a natural volatility proxy such as the spot volatility or the variance swap. Within this literature, two main approaches can be distinguished. The first works directly with the pricing PDE, perturbing around a specific model parameter as in \cite{lewis2000option,hagan2002managing,fouque2000derivatives,fouque2}. The second is probabilistic in nature, expressing option prices in terms of the joint law of the asset price and the variance swap \cite{antonelli2006pricing,fukasawa2011asymptotic,BergomiGuyon2012,alos2012decomposition,alos2020exponentiation}; this approach has the advantage of extending naturally to settings where volatility fails to be Markovian, such as rough volatility models. Within the Bachelier setting specifically, see \cite{baviera2025smile,unkwnown,alos2025short} and the references therein.

Across this literature, a common structure emerges: the expansion's first correction accounts for correlation effects via the leverage swap, the second captures vol-of-vol effects through the quadratic variation of the variance swap, with further refinements appearing at higher orders. While such expansions are accurate near the at-the-money point, they typically lack a closed analytical form \cite{Lewis02092022}, which in turn restricts the range of strikes over which they remain reliable.

The purpose of this paper is to obtain a two-scale analytical expansion for the Bachelier model using Taylor expansions and the Itô formula. The derivation of the expansion for the uncorrelated case has already been carried out in \cite{alos2026analytic}; here we generalize that result to the correlated case. When the expansions in both scales are truncated to a given order, the resulting approximation can be rearranged as a product of vectors and matrices, yielding a linear-algebra-based numerical method for option pricing and for the computation of the Delta and the Gamma of the option. Unlike classical Monte Carlo option pricing, where one expectation must be computed for each strike, our method allows option prices to be computed for as many strikes as desired by pre-computing a finite number of expectations, independent of the number of strikes.

 We analyse the convergence of the numerical method for the SABR model and we conduct numerical experiments under the SABR model and the rough Bergomi model in order to show the validity of our method.

\section{Preliminaries}
We consider the Bachelier model for asset prices under a risk-neutral probability $P$:
\begin{equation}
\label{Bachelier}
dX_t=\sigma _{t}\left( \rho dW_{t}+\sqrt{1-\rho ^{2}}%
B_{t}\right) ,\,t\in \lbrack 0,T] 
\end{equation}
for some $T>0$, where $W$ and $B$ are independent standard Brownian motions, $\rho \in \left[-1,1\right],$ and $\sigma$ is a square integrable
process adapted to the filtration generated by the Brownian motion $%
W$. As in the previous chapters, we denote by $\mathcal{F}^W$ and $\mathcal{F}^B$ the filtrations generated by $W$ and $ B$, respectively, and 
$\mathcal{F}:=\mathcal{F}^W\vee\mathcal{F}^B$. If $\sigma$ is constant and $\rho=0$, the above model is called the {\bf Bachelier model}.

We denote by $Bac(T,x,k,\sigma)$  the classical Bachelier price of a European call with time to maturity \(T\), current stock price \(x\),
strike price \(k\) and volatility \(\sigma\). That is, $$
  Bac(T,x,k,\sigma)=(x-k)N(d_{Bac}(k,\sigma))+N'(d_{Bac}(k,\sigma))\sigma\sqrt{T},$$
with
$$
  d_{Bac}(k,\sigma)=\frac{x-k}{\sigma\sqrt{T}},$$ where \(N\) is the cumulative distribution
function and the probability density  function of the standard normal random variable. 

\vspace{0.5cm}

We denote by $\mathcal{L}_{Bac}\left( \sigma \right) $ denotes the Bachelier 
differential operator with volatility $\sigma :$%
\begin{equation*}
\mathcal{L}_{Bac}\left( \sigma \right) =\frac{\partial }{\partial t}+\frac{1}{%
2}\sigma ^{2}\frac{\partial ^{2}}{\partial x^{2}}
\end{equation*}
It is well known that $\mathcal{L}_{Bac}\left( \sigma \right) Bac\left( \cdot
,\cdot,\cdot ;\sigma \right) =0.$

\vspace{0.5cm}

Finally, we define the Bachelier implied volatility of a traded call option $I^{Bac}(k)$  as the unique volatility parameter one should put in the Bachelier formula to get the market option price $V$. That is, the quantity $I^{Bac}(k)$ such that
$$
V=Bac(T,X_{0},k,I^{Bac}(k)),
$$
where $X_0$ denotes the asset price and $k$ the strike price of the option. Notice that, if $k=X_0$,
\begin{equation}
\label{Ibac}
V=Bac(T,X_{0},X_0,I^{Bac}(X_0))=N'(0)I^{Bac}\sqrt{T}=\frac{1}{\sqrt{2\pi}}I^{Bac}(X_0)\sqrt{T}.
\end{equation}
At the same time, due to the definition of the Black-Scholes implied volatility,
\begin{equation}
\label{Ibs}
V=BS(T,X_0,X_0,I(X_0))=X_0\left(2N\left(  \frac{I(X_0)\sqrt{T}}{2} \right)-1\right).
\end{equation}
Then, (\ref{Ibac}) and (\ref{Ibs}) imply the following conversion formula for ATM implied volatilities:
\begin{equation}
\label{Ibacbsatm}
I^{Bac}(X_0)=\frac{\sqrt{2\pi}}{\sqrt{T}}X_0\left(2N\left(  \frac{I(X_0)\sqrt{T}}{2} \right)-1\right)
\end{equation}
We will also need the following notations. 
\begin{itemize}
\item $v=\sqrt{\frac{1}{T}E\int_0^T\sigma_s^2 ds} $ is the square root of the variance swap.
\item $\hat{v}=E\sqrt{\frac{1}{T}\int_0^T\sigma_s^2 ds} $ is the volatility swap.
\item For all $s\in [0,T]$, we define $M_s=\frac{1}{T}E_s\int_0^T\sigma_u^2ds$.
\item For all $s\in [0,T]$, we denote $v_{s}=\sqrt{\frac{1}{T}E_s\int_0^T\sigma_u^2 du}$. In particular, $v_{0}=v.$
\item For all $s \in [0,T]$, we denote $\xi_s = \int_0^s \sigma_u dW_u$.
\end{itemize}
Notice that $v=\sqrt{M_0}$, $v_s=\sqrt{M_s}$, and $\hat{v}=E\sqrt{M_T}$. Then, a direct application of It\^o's formula to the process $M$ and the function $f(x)=\sqrt{x}$ leads to the following relationship between the variance and the volatility swap
\begin{equation}\label{convexity}
\hat{v}=v-\frac{1}{8}E\int_0^T\frac{1}{v_{s}^3}d\langle M,M\rangle_s.
\end{equation}
We also recall the Hull-White formula, that allows us to write the Bachelier option price $V$ as
\[
V = E\left[ Bac(T,X_0 + \rho \xi_T, k, \sqrt{1-\rho^2}v_T)\right].
\]
We also define the Delta ($\Delta$) and the Gamma ($\Gamma$) of the option $V$ as
\[
\Delta = \partial_{X_0}V, \quad \Gamma = \partial^2_{X_0} V.
\]
In this paper we will show a numerical method that allows us to compute $V$,$\Delta$ and $\Gamma$ using basic linear algebra. To do so, we will first recall some linear algebra definitions in order to state the notation we will use from now on.
\begin{defi}
    We will denote by $I_n \in \R^{n \times n}$ the identity matrix of order $n$.
\end{defi}
\begin{defi}
    Let $A = (a_{i,j}) \in \R^{m,n}$, $B = (b_{i,j}) \in \R^{p,q}$. We denote the Kronecker product of $A$ times $B$ as the block matrix defined by
    \[
    A \otimes B = \begin{pmatrix}
        a_{11}B & a_{12}B & \dots & a_{1n}B \\
        a_{21}B & a_{22} B & \dots & a_{2n}B \\
        \vdots & \vdots & \dots & \vdots \\
        a_{m1}B & a_{m2}B & \dots & a_{mn}B
    \end{pmatrix}.
    \]
\end{defi}
\section{An algebraic option pricing method}
In this section we present a numerical method for computing option prices under the Bachelier model, as well as the $\Delta$ and the $\Gamma$ of the option. 

\begin{methodbox}{Method for computing $V$, $\Delta$ and $\Gamma$ separately}
    Choose $M_{max}, N_{max} \in \N$. We can compute $V$ as
    \[
V = a_0 + \Vec{\rho}_e \mathcal{M}_e \mathbf{x}_e + \Vec{\rho}_o \mathcal{M}_o \mathbf{x}_o,
\]
where 
\[
a_0 = Bac\left(X_0,k,\sqrt{1-\rho^2}\,v\right), 
\]
 \[
    \Vec{\rho}_e^T = \begin{pmatrix}
        1 \\
        \frac{\rho^2}{2} \\
        \frac{\rho^4}{4!} \\
        \vdots
    \end{pmatrix}, \quad \mathcal{M}_e = \begin{pmatrix}
        b_0 & b_1 & \dots \\
        b(0,1) & b(0,2) & \dots \\
        \vdots & \vdots & \vdots & \\
        b(N_{max}, 1) & b(N_{max}, 2) & \dots
    \end{pmatrix}, \quad  \mathbf{x}_e = \begin{pmatrix}
        1  \\
        (X_0 -k)^2 \\
        (X_0 - k)^4 \\
        \vdots \\
    \end{pmatrix},
    \]
     \[
    \Vec{\rho}_o^T = \begin{pmatrix}
        \rho \\
        \frac{\rho^3}{3!} \\
        \frac{\rho^5}{5!} \\
        \vdots
    \end{pmatrix}, \quad \mathcal{M}_o = \begin{pmatrix}
        c_1 & c_2 & \dots \\
        c(0,1) & c(0,2) &\dots \\
        \vdots & \vdots& \vdots \\
        c(N_{max}, 1) & c(N_{max}, 2) & \dots
    \end{pmatrix}, \quad \mathbf{x}_o = \begin{pmatrix}
        (X_0-k)  \\
        (X_0 -k)^3 \\
        (X_0 - k)^5 \\
        \vdots \\
    \end{pmatrix},
    \]
and
\[
b_n = \frac{-\sqrt{T}}{\sqrt{2\pi}}\frac{(-1)^n\left(E(v_T^{1-2n}) - v^{1-2n} \right)}{n!(2n-1) 2^n T^n(1-\rho^2)^{n-\frac{1}{2}}},
\]
\[
 c_n =\frac{1}{2\sqrt{2\pi}}\frac{(-1)^nE\left[\xi_T v_T^{-1-2n} \right]}{ n! 2^n T^{n+1/2}(1-\rho^2)^{n+1/2}} \frac{4}{4n+2},
\]
\[
b(p,n) = \frac{1}{\sqrt{2\pi}}\frac{(-1)^{n+p}(2n+2p)!}{(2n)! (n+p)!2^{n+p}(1-\rho^2)^{n+p+1/2}T^{n+p+1/2}}E\left[ \xi_T^{2p+2} v_T^{-2n-2p-1}\right],
\]
and
\[
c(p,n) = \frac{-1}{\sqrt{2\pi}} \frac{(-1)^{n+p} (2n+2p+2)!}{(2n+1)! (n+p+1)! 2^{n+p+1} (1-\rho^2)^{n+p+3/2}T^{n+p+3/2}} E[\xi_T^{2p+3}v_T^{-2n-2p-3}].
\]
The $\Delta$ and the $\Gamma$ are computed respectively as
  \[
    \Delta = \partial_{X_0}a_0 + \Vec{\rho}_e \mathcal{M}_e \partial_{X_0}\mathbf{x}_e + \Vec{\rho}_o \mathcal{M}_o \partial_{X_0}\mathbf{x}_o,
    \]
    \[
    \Gamma = \partial^2_{X_0}a_0 + \Vec{\rho}_e  \mathcal{M}_e \partial^2_{X_0}\mathbf{x}_e + \Vec{\rho}_o \mathcal{M}_o \partial^2_{X_0}\mathbf{x}_o.
    \]
\end{methodbox}
\begin{remark}
    The method presented in this paper can be slightly modified in order to compute $V$, $\Delta$ and $\Gamma$ simmultaniusly. Indeed, under the same definitions, 
\[
\begin{pmatrix}
    V \\
    \Delta \\
    \Gamma
\end{pmatrix} = \begin{pmatrix}
    Bac(T, X_0, k, \sqrt{1-\rho^2}v) \\
    \Phi\left(d(k, \sqrt{1-\rho^2}v) \right) \\
    \frac{\phi(d(k, \sqrt{1-\rho^2 v)}}{\sqrt{1-\rho^2}v \sqrt{T}}
\end{pmatrix} + \mathbf{\rho}^T \begin{pmatrix}
    \mathcal{M}\textbf{x} \\
    \mathcal{M}D \textbf{x} \\
    \mathcal{M}D^2 \textbf{x}
\end{pmatrix}.
\]
Here,
\[
\mathbf{\rho} = (\rho^m)_{m=0}^{M_{max}},
\]
\[
\mathcal{M}_{i,j} = \begin{cases}
    b_{\frac{j-1}{2}} & i = 0, j \text{ odd}, \\
    \frac{1}{(i-1)!}b\left(\frac{i-3}{2}, \frac{j+1}{2} \right) & i \geq 3 \text{ odd}, j \text{ odd}, \\
    c_{\frac{j}{2}} & i = 2, j \text{ even} \\
    \frac{1}{(i-1)!}c\left( \frac{i-4}{2}, \frac{j}{2}\right) & i \geq 4 \text{ even}, j \text{ even}, \\
    0 & i+j  \text{ even}
\end{cases}
\]
and the matrix $D$ is defined as $D_{j,j-1} = j-1$ and $D_{i,j} = 0$ otherwise. For this representation, the vector $\mathbf{x}$ contains all powers of $X_0-k$, that is,
\[
\mathbf{x} = (1,X_0-k, (X_0-k)^2, \dots)^T.
\]
\end{remark}
\begin{remark}
    The method is stated for a single strike. In order to compute option prices and Greeks for strikes $k_1, \dots, k_r$ simply replace the vector $\mathbf{x}$ by the matrix
    \[
    \mathbf{x} = \begin{pmatrix}
        1 & 1 & \dots & 1 \\
    X_0-k_1 & X_0-k_2 & \dots& X_0 -k_r \\
    (X_0-k_1)^2 & (X_0 - k_2)^2& \dots & (X_0- k_r)^2 \\
    \vdots & \vdots & & \vdots
    \end{pmatrix}.
    \]
\end{remark}
\begin{remark}
    Notice that the number of expectations one needs to compute in order to run the method depends only on the hyper-parameters $M_{max}$ and $N_{max}$, but not on the number of strikes $r$. This means that, by computing only a finite number of expectations (depending on $M_{max}$ and $N_{max}$), we can obtain option prices and Greeks for an arbitrary number of strikes within a convergence domain that depends on the specific stochastic volatility model under consideration. In Section \ref{s: convergence SABR}, we derive such a range of convergence for the SABR model.
\end{remark}
The rest of the paper is devoted to show the derivation of the numerical method, as well as a convergence analysis for the SABR model. Finally we will show the accuracy of the method using numerical simulation.
\subsection{Derivation of the method}
The key for obtaining the method is expanding $V$ using Taylor expansions and the Itô formula in a specific manner. Firstly, a classical Taylor expansion around the spot price argument allows us to write the following representation for $V$.
\begin{prop} \label{prop expansion rho}
    The option price $V$ can be expressed as
    \begin{equation} \label{Expansion rho}
    V = \sum_{m=0}^{\infty} a_m(\rho)\rho^m,
    \end{equation}
    where
    \[
    a_m(\rho) = \frac{1}{m!}E\left[ \partial_x^mBac(T,X_0, k, \sqrt{1-\rho^2}v_T)\xi_T^m\right]
    \]
    provided the right hand side of \eqref{Expansion rho} is convergent.
\end{prop}
\begin{proof}
    The Hull-White formula states that
    \[
    V =  E\left[ Bac(T,X_0 + \rho \xi_T, k, \sqrt{1-\rho^2}v_T)\right].
    \]
    A Taylor expansion around the spot argument under the assumption that the sum and the expectation can be exchanged leads to
    \[
    V = \sum_{m=0}^{\infty} \frac{1}{m!}E \left[  \partial_x^mBac(T,X_0, k, \sqrt{1-\rho^2}v_T)(\rho \xi_T)^m\right].
    \]
    Rearranging the constants we obtain the desired representation.
\end{proof}
To obtain a desirable representation for each coefficient $a_m(\rho)$, we will apply the Itô formula to the process $\partial_x^mBac(T,X_0, k, \sqrt{1-\rho^2}v_T)\xi_T^m$ in order to obtain a similar representation as in \cite{alos2026analytic}. Before obtaining the representation, we need some auxiliary results regarding the derivatives of the Bachelier function.
\begin{lema}
    Let $p \in \N \cup \{0\}$. Then, for every $\sigma > 0$ we have
    \[
    \partial_x^{2p+2}Bac(T,X_0,K,\sigma)=\frac{1}{\sqrt{2\pi}}\sum_{n=0}^{\infty}\frac{(-1)^{n+p}(2n+2p)!}{(2n)!(n+p)!2^{\,n+p}}\,\frac{(X_0-K)^{2n}}{\sigma^{2n+2p+1}T^{n+p+1/2}},
    \]
    and
    \[
    \partial_x^{2p+3}Bac(T,X_0, k, \sigma) = \frac{1}{\sqrt{2\pi}}\sum_{n=0}^{\infty}\frac{(-1)^{n+p+1}(2n+2p+2)!}{(2n+1)!(n+p+1)!2^{n+p+1}}\frac{(X_0-K)^{2n+1}}{\sigma^{2n+2p+3}T^{n+p+3/2}}.
    \]
\end{lema}
\begin{proof}
    The Bachelier function satisfies
    \[
    \partial^2_x Bac(T,X_0, k, \sigma) = \frac{1}{\sigma\sqrt{T}}N'(d)
    \]
    and, since $\partial_x d = \frac{1}{\sigma\sqrt{T}}$ all the derivatives of order $m \geq 2$ of the Bachelier function can be written as
    \[
    \partial^m_x Bac(T,X_0, k, \sigma) = \frac{1}{(\sigma \sqrt{T})^{m-1}} N^{(m-1)}(d).
    \]
    Now, the function $N'$ is even in $d$, so its power series representation is written as
    \[
    N'(d) = \frac{1}{\sqrt{2\pi}}\sum_{l = 0}^{\infty} \frac{(-1)^l}{2^l l!} d^{2l}.
    \]
    Computing the derivatives of $N'$ via its power series we find that
    \[
    (N')^{(2p)}(d) = \frac{1}{\sqrt{2\pi}}\sum_{n =0}^{\infty} \frac{(-1)^{n+p}(2n+2p)!}{2^{n+p}(n+p)!(2n)!}d^{2n},
    \]
    and
    \[
    (N')^{(2p+1)} = \frac{1}{\sqrt{2\pi}}\sum_{n=0}^{\infty} \frac{(-1)^{n+p+1}(2n+2p+2)!}{2^{n+p+1}(n+p+1)!(2n+1)!}d^{2n+1}.
    \]
    Substituting
    \[
    d = \frac{X_0 - k}{\sigma \sqrt{T}}
    \]
    we obtain the desired result.
\end{proof}
The application of the Itô formula to the process $\partial_x^mBac(T,X_0, k, \sqrt{1-\rho^2}v_T)\xi_T^m$ allows us to express $a_m(\rho)$ as power series, provided those power series are convergent. 
\begin{teo} \label{coefs en rho}
    Let $a_m(\rho)$ be defined as in Proposition \ref{prop expansion rho}. Then, the following representation for $a_m(\rho)$ holds:
    \begin{itemize}
        \item If $m = 0$,
         \begin{align*}
        a_0(\rho) =&Bac\left(X_0,k,\sqrt{1-\rho^2}\,v\right) + \frac{\sqrt{1-\rho^2}\sqrt{T}}{\sqrt{2\pi}}\big(\hat v - v\big)\\
        - &\frac{\sqrt{T}}{\sqrt{2\pi}}\sum_{n=1}^{\infty} \frac{(-1)^n(X_0-k)^{2n}}{n!\,(2n-1)}\left(\frac{1}{2T}\right)^n (1-\rho^2)^{\frac{1}{2}-n}\Big(E[v_T^{1-2n}]-v^{1-2n}\Big).
    \end{align*}
    \item If $m = 1$,
    \[
    a_1(\rho) = \frac{1}{2\sqrt{2\pi}}\sum_{n=0}^{\infty}\frac{(-1)^n(X_0- k)^{2n+1}}{ n! 2^n T^{n+1/2}(1-\rho^2)^{n+1/2}} \frac{4}{4n+2}E\left[\xi_T v_T^{-1-2n} \right].
    \]
    \item If $m = 2p+2$ with $p \in \N \cup \{0\}$, then
    \begin{align*}
   & a_{2p+2}(\rho) \\
   = &\frac{1}{(2p+2)!\sqrt{2\pi}}\sum_{n=0}^{\infty} \frac{(-1)^{n+k}(2n+2p)!(X_0 -k)^{2n}}{(2n)! (n+p)!2^{n+p}(1-\rho^2)^{n+p+1/2}T^{n+p+1/2}}E\left[ \xi_T^{2p+2} v_T^{-2n-2p-1}\right].
    \end{align*}
    \item If $m = 2p+3$ with $p \in \N \cup \{0\}$, then
    \begin{align*}
    &a_{2p+3}(\rho) \\
    = &\frac{-1}{(2+3)!\sqrt{2\pi}}\sum_{n=0}^{\infty}\frac{(-1)^{n+p} (2n+2p+2)!(X_0-k)^{2n+1}}{(2n+1)! (n+p+1)! 2^{n+p+1} (1-\rho^2)^{n+p+3/2}T^{n+p+3/2}} E[\xi_T^{2p+3}v_T^{-2n-2p-3}],
    \end{align*}
    \end{itemize}
    provided all series are convergent.
\end{teo}
\begin{proof}
    Recall that 
    \[
    m!a_m(\rho) = E\left[ \partial_x^m Bac(T,X_0, k, \sqrt{1-\rho^2}v_T) \xi_T^m\right]. 
    \]
    Using the fact that both processes $M_s = v_s^2$ and $\xi_s$ are martingales, we can apply the Itô formula to the process $\partial_x^m Bac(T,X_0, k, \sqrt{1-\rho^2}v_T) \xi_T^m$ in order to obtain
    \begin{align*}
           E\left[\partial_x^m Bac(T,X_0, k, \sqrt{1-\rho^2}v_T) \xi_T^m\right] 
       = &E\left[ \partial_x^m Bac(T,X_0, k, \sqrt{1-\rho^2}v_0) \xi_0^m\right] \\
        + &\frac{m(m-1)}{2}E\left[ \int_0^T \xi_s^{m-2} \sigma_s^2 \partial_x^m Bac(T,X_0, k, \sqrt{1-\rho^2} v_s) ds\right] \\
    + &\frac{T^2(1-\rho^2)^2}{8}E\left[\int_0^T \xi_s^m \partial_x^{m+4}Bac(T,X_0, k, \sqrt{1-\rho^2}v_s) d\langle M, M \rangle_s \right] \\
    + &\frac{mT(1-\rho^2)}{2}E \left[\int_0^T \partial_x^{m+2}Bac(T,X_0, k, \sqrt{1-\rho^2}v_s) \xi_s^{m-1} \sigma_s d\langle W, M \rangle_s \right].
    \end{align*}
    Notice that not all terms contribute for all $m \geq 0$. Indeed, the terms $m = 0$, $m = 1$ and $m \geq 2$ have to be treated separately.
    \begin{itemize}
        \item For $m = 0$ the result of the Itô formula yields
        \begin{align*}
        E[Bac(T,X_0,k, \sqrt{1-\rho^2}v_T)] = 
        &E[Bac(T,X_0,k, \sqrt{1-\rho^2}v_0)] \\ 
        + &\frac{T^2(1-\rho ^2)^2}{8} E\left[ \int_0^T \partial_x^4 Bac(T,X_0, k, \sqrt{1-\rho^2}v_s) d\langle M, M\rangle_s\right].
        \end{align*}
       A Taylor expansion around $d = \frac{X_0 - k}{\sqrt{T}\sqrt{1-\rho^2}v_s}$ shows that $a_0(\rho)$ can be written as
       \begin{align*}
      &Bac(T,X_0,k, \sqrt{1-\rho^2}v_0)  \\
        - &\frac{T^2(1-\rho ^2)^2}{8} E\left[ \int_0^T \frac{1}{\sqrt{2\pi}}\sum_{n=0}^{\infty} \frac{(-1)^n (X_0-k)^{2n}}{(2n)! (n+1)! 2^{n+1}T^{n+\frac{3}{2}} (1-\rho^2)^{n+\frac{3}{2}}v_s^{2n+3}} d\langle M, M\rangle_s\right] \\
        =& Bac(T,X_0,k, \sqrt{1-\rho^2}v_0)  \\
        - &\frac{T^2(1-\rho^2)^2}{8\sqrt{2\pi}}\sum_{n=0}^{\infty} \frac{(-1)^n (X_0-k)^{2n}}{(2n)! (n+1)! 2^{n+1}T^{n+\frac{3}{2}}(1-\rho^2)^{n+\frac{3}{2}}}E\left[\int_0^T \frac{1}{v_s^{2n+3}}d\langle M, M \rangle_s \right].
        \end{align*}
        Using that
        \[
        E\left[\int_0^T \frac{1}{v_s^{2n+3}} d\langle M, M \rangle_s\right] = \frac{8}{4n^2 - 1}\left( E[v_T^{1-2n}] - v^{1-2n}\right).
        \]
        Rearranging the terms of the expansion we are lead to
        \begin{align*}
            a_0(\rho) = &Bac\left(T,X_0, k , \sqrt{1-\rho^2}v\right) + \frac{\sqrt{1-\rho^2}\sqrt{T}}{\sqrt{2\pi}}(\hat{v} - v) \\ 
           -&  \frac{\sqrt{T}}{\sqrt{2\pi}}\sum_{n=1}^{\infty} \frac{(-1)^n(X_0-k)^{2n}}{n!(2n-1)2^n T^n(1-\rho^2)^{n-\frac{1}{2}}}\left( E[v_T^{1-2n}] - v^{1-2n}\right).
        \end{align*}
        \item For $m = 1$ the Itô formula yields
        \begin{align*}
       E\left[ Bac(T,X_0, k, \sqrt{1-\rho^2}v_T) \xi_T\right]  = &\frac{T^2(1-\rho^2)^2}{8}E\left[\int_0^T \xi_s\partial_x^{5}Bac(X_0, k, \sqrt{1-\rho^2}v_s) d\langle M, M \rangle_s \right]\\
        + &\frac{T(1-\rho)}{2}E\left[\int_0^T \partial_x^3 Bac(X_0, k, \sqrt{1-\rho^2}v_s) \sigma_s d\langle W, M\rangle_s \right].
    \end{align*}
    We will apply a Taylor expansion to the two terms on the right hand side of the previous equation and we will later merge them using the Itô formula.
    The first term yields
    \begin{align*}
    &\frac{T^2(1-\rho^2)^2}{8}E\left[\int_0^T \xi_s \partial_x^5 Bac(X_0, k, \sqrt{1-\rho^2}v_s) d\langle M, M \rangle_s \right] \\
    = &\frac{T(X_0-k)}{2\sqrt{2\pi}}\sum_{n=0}^{\infty} \frac{(-1)^n (X_0-k)^{2n}}{n! 2^n T^{n+3/2}(1-\rho^2)^{n+1/2}} \frac{2n+3}{4} E\left[ \int_0^T \frac{\xi_s}{v_s^{2n+5}} d\langle M, M \rangle_s\right]
    \end{align*}
    while the second term satisfies
    \begin{align*}
        &\frac{T(1-\rho^2)}{2}E\left[\int_0^T \partial_x^3 Bac(X_0, k, \sqrt{1-\rho^2}v_s) \sigma_s d\langle W, M\rangle_s \right] \\
        = &\frac{-T(X_0-k)}{2\sqrt{2\pi}}\sum_{n=0}^{\infty} \frac{(-1)^n (X_0-k)^{2n}}{n!2^n T^{n+3/2}(1-\rho^2)^{n+1/2}} E\left[\int_0^T \frac{\sigma_s}{v_s^{2n+3}} d\langle W, M \rangle_s\right].
    \end{align*}
    Adding both Taylor expansions we obtain
    \[
    \frac{T(X_0-k)}{2\sqrt{2\pi}}\sum_{n=0}^{\infty} \frac{(-1)^n (X_0-k)^{2n}}{n! 2^n T^{n+3/2}(1-\rho^2)^{n+1/2}} \left(\frac{2n+3}{4} E\left[ \int_0^T \frac{\xi_s}{v_s^{2n+5}} d\langle M, M \rangle_s\right] - E\left[ \int_0^T \frac{\sigma_s}{v_s^{2n+3}}d\langle W, M\rangle_s\right] \right).
    \]
     Applying the Itô formula to $\xi_T v_T^{\alpha}$ leads us to
    \[
    E[\xi_T v_T^{2\alpha}]= \frac{\alpha (\alpha-1)}{2}E\left[ \int_0^T \frac{\xi_s}{v_s^{4-2\alpha}} d \langle M, M\rangle_s\right] +\alpha E\left[ \int_0^T \frac{1}{v_s^{2-2\alpha}}d\langle \xi, M\rangle_s\right],
    \]
    where we have used that $\sigma_sd\langle W, M \rangle_s = d\langle \xi, M \rangle_s$. Choosing $2\alpha = -1-2n$ leads to
    \[
    \frac{4}{4n+2}E\left[ \xi_T v_T^{-1-2n}\right] = \frac{2n+3}{4}E\left[ \int_0^T \frac{\xi_s}{v_s^{2n+5}} d\langle M, M \rangle_s\right] - E\left[ \int_0^T \frac{1}{v_s^{2n+3}}d\langle \xi, M \rangle_s\right]. 
    \]
    Thus, the term $a_1(\rho)$ can be written as
     \[
     \frac{T(X_0-k)}{2\sqrt{2\pi}}\sum_{n=0}^{\infty} \frac{(-1)^n (X_0-k)^{2n}}{n! 2^n T^{n+3/2}(1-\rho^2)^{n+1/2}} \frac{4}{4n+2}E[\xi_T v_T^{-1-2n}].
    \]
    \item For $m \geq 2$ the Itô formula leads us to
     \begin{align*}
           E\left[\partial_x^m Bac(T,X_0, k, \sqrt{1-\rho^2}v_T) \xi_T^m\right] 
       =  &\frac{m(m-1)}{2}E\left[ \int_0^T \xi_s^{m-2} \sigma_s^2 \partial_x^m Bac(T,X_0, k, \sqrt{1-\rho^2} v_s) ds\right] \\
    + &\frac{T^2(1-\rho^2)^2}{8}E\left[\int_0^T \xi_s^m \partial_x^{m+4}Bac(T,X_0, k, \sqrt{1-\rho^2}v_s) d\langle M, M \rangle_s \right] \\
    + &\frac{mT(1-\rho^2)}{2}E \left[\int_0^T \partial_x^{m+2}Bac(T,X_0, k, \sqrt{1-\rho^2}v_s) \xi_s^{m-1} \sigma_s d\langle W, M \rangle_s \right].
    \end{align*}
    The strategy is the same, we expand each term using a Taylor expansion, we merge the three series and we try to simplify the stochastic terms using the Itô formula. We will only detail the case $m = 2p+2$ since the case $m = 2p+3$ follows the same argument. For $m = 2p+2$, the term $a_{2p+2}(\rho)$ can be written after the Taylor expansion as
    \[
\frac{1}{\sqrt{2\pi}} \sum_{n=0}^{\infty} \frac{(-1)^{n+p}(X_0-k)^{2n}}{n! 2^{2n+p}(1-\rho^2)^{n+p+1/2}T^{n+p+1/2}}A_{n,p},
\]
where
\begin{align*}
    A_{n,p} = &\frac{(2p+1)(p+1)(2n+2p)!}{(n+p)!}E\left[\int_0^T \frac{\xi^{2p}}{v_s^{2n+2p+1}} d  \langle \xi, \xi\rangle_s\right] \\
    - & \frac{(2p+2)(2n+2p+2)!}{4(n+p+1)!}E\left[ \int_0^T \frac{\xi^{2p+1}_s}{v_s^{2n+2p+3}} d\langle \xi, M\rangle_s\right] \\
    + & \frac{(2n+2p+4)!}{8(n+p+2)!} E\left[\int_0^T\frac{\xi_s^{2p+2}}{v_s^{2n+2p+5}} d\langle M, M\rangle_s \right].
\end{align*}
A direct application of the Itô formula shows that
\[
A_{n,p} = \frac{(2n+2p)!}{(n+p)!} E\left[ \xi_T^{2p+2} v_T^{-2n-2p-1}\right]
\]
and therefore the term $a_{2p+2}(\rho)$ can be written as
\[
\frac{1}{(2p+2)!\sqrt{2\pi}} \sum_{n=0}^{\infty}\frac{(-1)^{n+p}(2n+2p)!(X_0-k)^{2n}}{(2n)! (n+p)!2^{n+p}(1-\rho^2)^{n+p+1/2}T^{n+p+1/2}}E\left[ \xi_T^{2p+2} v_T^{-2n-2p-1}\right].
\]
The same argument for $m = 2p+3$ shows that $a_{2p+3}(\rho)$ can be written as
\[
\frac{-1}{(2p+3)!\sqrt{2\pi}}\sum_{n=0}^{\infty} \frac{(-1)^{n+p} (2n+2p+2)!(X_0-k)^{2n+1}}{(2n+1)! (n+p+1)! 2^{n+p+1} (1-\rho^2)^{n+p+3/2}T^{n+p+3/2}} E[\xi_T^{2p+3}v_T^{-2n-2p-3}].
\]
    \end{itemize}
\end{proof}
\begin{remark}
    The coefficient $a_0(\rho)$ evaluated at $\rho = 0$ coincides with the analytical expansion found in \cite{alos2026analytic}, showing that Theorem \ref{coefs en rho} in fact extends the uncorrelated case to broader instances of stochastic volatility Bachelier-type models.
\end{remark}
The derivation of the numerical method is now a matter of rewriting the conclusion of Theorem \ref{coefs en rho}. Indeed, define
\[
b_n = \frac{-\sqrt{T}}{\sqrt{2\pi}}\frac{(-1)^n\left(E(v_T^{1-2n}) - v^{1-2n} \right)}{n!(2n-1) 2^n T^n(1-\rho^2)^{n-\frac{1}{2}}},
\]
\[
 c_n =\frac{1}{2\sqrt{2\pi}}\frac{(-1)^nE\left[\xi_T v_T^{-1-2n} \right]}{ n! 2^n T^{n+1/2}(1-\rho^2)^{n+1/2}} \frac{4}{4n+2},
\]
\[
b(p,n) = \frac{1}{\sqrt{2\pi}}\frac{(-1)^{n+p}(2n+2p)!}{(2n)! (n+p)!2^{n+p}(1-\rho^2)^{n+p+1/2}T^{n+p+1/2}}E\left[ \xi_T^{2p+2} v_T^{-2n-2p-1}\right],
\]
and
\[
c(p,n) = \frac{-1}{\sqrt{2\pi}} \frac{(-1)^{n+p} (2n+2p+2)!}{(2n+1)! (n+p+1)! 2^{n+p+1} (1-\rho^2)^{n+p+3/2}T^{n+p+3/2}} E[\xi_T^{2p+3}v_T^{-2n-2p-3}].
\]
Theorem \ref{coefs en rho} shows that, if we assume convergence of the series, we can write the coefficients of the expansion as
\[
a_0(\rho) = Bac(T,X_0, k, \sqrt{1-\rho^2}v) + \sum_{n=0}^{\infty} b_n(X_0-k)^{2n},
\]
\[
a_1(\rho) = \sum_{n=0}^{\infty} c_n (X_0 - k)^{2n+1},
\]
\[
a_{2p+2}(\rho) = \sum_{n=0}^{\infty} \frac{1}{(2p+2)!}b(p,n) (X_0 - k)^{2n},
\]
\[
a_{2p+3}(\rho) = \sum_{n=0}^{\infty}\frac{1}{(2p+3)!} c(p,n) (X_0-k)^{2n+1}.
\]
 Assume $M_{max}$ is even. We can split $V^{M_{max}}$ as a sum of the form
\[
V^M = \left(a_0(\rho) + a_2(\rho)\rho^2 + \dots + a_{M_{max}}(\rho)\rho^{M_{max}}\right) + (a_1(\rho)\rho + a_3(\rho)\rho^3 + \dots + a_{M_{max}-1}(\rho)\rho^{M_{max}-1})
\]
Each sum can be expressed using linear algebra. Notice that the sum over the even powers of $\rho$ leads to
\[
Bac(T,X_0, k, \sqrt{1-\rho^2}v) +  \Vec{\rho}_e \mathcal{M}_e \mathbf{x}_e,
\]
where
 \[
    \Vec{\rho}_e^T = \begin{pmatrix}
        1 \\
        \frac{\rho^2}{2} \\
        \frac{\rho^4}{4!} \\
        \vdots
    \end{pmatrix}, \quad \mathcal{M}_e = \begin{pmatrix}
        b_0 & b_1 & \dots \\
        b(0,1) & b(0,2) & \dots \\
        \vdots & \vdots & \vdots & \\
        b(N_{max}, 1) & b(N_{max}, 2) & \dots
    \end{pmatrix}, \quad  \mathbf{x}_e = \begin{pmatrix}
        1  \\
        (X_0 -k)^2 \\
        (X_0 - k)^4 \\
        \vdots \\
    \end{pmatrix}.
    \]
In the same way, the sum over the odd powers of $\rho$ leads to
\[
\Vec{\rho}_o \mathcal{M}_o \mathbf{x}_o,
\]
where
\[
 \Vec{\rho}_o^T = \begin{pmatrix}
        \rho \\
        \frac{\rho^3}{3!} \\
        \frac{\rho^5}{5!} \\
        \vdots
    \end{pmatrix}, \quad \mathcal{M}_o = \begin{pmatrix}
        c_1 & c_2 & \dots \\
        c(0,1) & c(0,2) &\dots \\
        \vdots & \vdots& \vdots \\
        c(N_{max}, 1) & c(N_{max}, 2) & \dots
    \end{pmatrix}, \quad \mathbf{x}_o = \begin{pmatrix}
        (X_0-k)  \\
        (X_0 -k)^3 \\
        (X_0 - k)^5 \\
        \vdots \\
    \end{pmatrix}.
\]
Adding the sum over the even and odd powers of $\rho$ we obtain
\[
V^{M_{max}} =Bac(T,X_0, k, \sqrt{1-\rho^2}v) +  \Vec{\rho}_e \mathcal{M}_e \mathbf{x}_e + \Vec{\rho}_o \mathcal{M}_o \mathbf{x}_o.
\]
\section{Convergence in the SABR model} \label{s: convergence SABR}
In this section we will discuss the convergence of the method under the SABR model. To do so, we will need first a general bound on the derivatives of the Bachelier function, that relies on Crámer's inequality for Hermite polynomials (see \cite{greengard1991fast}).
\begin{prop}
    The derivative of order $m$ of the Bachelier function $Bac$ satisfies
    \[
    \partial_x^m Bac(T,X_0, k, \sqrt{1-\rho^2}v_T) = \frac{(-1)^m}{(\sqrt{1-\rho^2}\sqrt{T}v_T)^{m-1}}H_{m-2}(d)N'(d),
    \]
    where
    \[
    d = \frac{X_0 - k}{\sqrt{1-\rho^2}\sqrt{T}v_T}.
    \]
    In particular,
    \[
    \frac{1}{m!}|\partial^m_x Bac(T,X_0, k, \sqrt{1-\rho^2}v_T)| \leq \frac{K}{\sqrt{2\pi}} \frac{1}{((1-\rho^2)T)^{\frac{m-1}{2}}}\frac{\sqrt{(m-1)!}}{\sqrt{m!}} v_T^{-(m-1)},
    \]
    where $K$ is the Crámer constant $K < 1.09$.
\end{prop}
 In order to obtain convergence results for our method, an analysis of the moments $E[\xi_T^{\beta}v_T^{-2\alpha}]$ is required. This analysis can be done using rough estimates, but the region of convergence obtained in this case is smaller than the actual range of convergence that can be seen computationally. Hence, if the joint distribution of $\xi_T$ and $v_T^2$ is known one can obtain sharper bounds and better convergence criterion. Assume model \eqref{Bachelier} with $\sigma$ satisfying
\[
\sigma_t = \sigma_0 \exp \left(\nu W_t - \frac{\nu^2 t}{2} \right).
\]
In the case of the SABR model, the joint distribution of $\xi_T$ and $v_T^2$ is well-known and it can be expressed by means of Hartman-Watson distribution(see \cite{hartman74normal} and \cite{yor92exponential}). The study on the distribution of the time integral of the Brownian motion has been applied for pricing options under the SABR model, for instance in \cite{cai2017exact} and \cite{pirjol2026conditional} among others.
\begin{defi}
    We define the Hartman-Watson density function $\theta(r,t)$ as
    \[
    \theta(r,t) = \frac{r}{\sqrt{2\pi^3t}} e^{\frac{\pi^2}{2t}} \int_0^{\infty} e^{\frac{-\eta^2}{2t}} e^{-r \cosh \eta} \sinh \eta \sin \frac{\pi \eta}{t} d \eta.
    \]
\end{defi}
\begin{prop}
    Let $A_t^{(\mu)}$ be defined as
    \[
    A_t^{(\mu)} = \int_0^t \exp(2(W_s + \mu s)) ds.
    \]
    Then, the joint distribution of the vector $(W_t + \mu t, A_t^{(\mu)})$ has density function
    \begin{align*}
    f(x,u) = &P(W_t + \mu t \in dx, A_t^{(\mu)} \in du)\\
    = &e^{\mu x - \frac{\mu^2}{2}t}\exp\left(-\frac{1+e^{2x}}{2u} \right) \theta\left( \frac{e^x}{u}, t\right) \frac{du dx}{t}.
    \end{align*}
\end{prop}
The numerical method we present is a two-scale expansion. Hence, the convergence needs to be studied for both scales. We will first show the convergence in $\rho$. In order to do so, since the series in $\rho$ is not a power series, we can not rely on the classical convergence criteria. To analyse the convergence we will use the dominated convergence theorem. First, using the bounds on the derivatives of the Bachelier function we find that
\begin{align*}
a_m(\rho) = &\frac{1}{m!}\left|E\left[ \partial_x^m Bac(T,X_0, k, \sqrt{1-\rho^2}v_T) \rho^m \xi_T^m\right]\right|\\
\leq &\frac{1}{\sqrt{2\pi}}\frac{1}{((1-\rho^2)T)^{\frac{m-1}{2}}}\frac{\sqrt{(m-2)!}}{m!}\rho^m |\xi_T|^m v_T^{-(m-1)}.
\end{align*}
We will analyse the mixed moments $E[|\xi_T|^{\beta}v_T^{-2\alpha}]$. The first reduction that can be done in the SABR case follows directly from the Itô formula.
\begin{lema}
    We have
    \[
    \xi_T = \frac{\sigma_T - \sigma_0}{\nu}.
    \]
\end{lema}
Thus, the mixed moments can be computed by means of the Hartman-Watson distribution.
\begin{prop}
    There exists a Brownian motion $\tilde{W}$ such that for $\mu = -1/2$ and $\tau = \nu^2 T$, the element
    \[
    A^{(\mu)}_{\tau} = \int_0^{\tau} \exp(2(\tilde{W}_s + \mu s)) ds
    \]
    satisfies
    \[
    v_T^2 = \frac{\sigma_0^2}{\nu^2 T}A_{\tau}^{(\mu)}, \quad \xi_T = \frac{\sigma_0}{\nu}\left(e^{\tilde{W}_{\tau} + \mu \tau} - 1 \right),
    \]
   where the equalities are taken in distribution.
\end{prop}
\begin{proof}
    Since $\sigma$ satisfies
    \[
    \sigma_t = \sigma_0 \exp \left(\nu W_t - \frac{\nu^2 t}{2} \right),
    \]
    we have
    \[
    \int_0^T \sigma_s^2 ds = \sigma_0^2\int_0^T \exp \left(2\nu W_s - \nu^2 s \right)ds.
    \]
    By the change of variables $u = \nu^2 s$, $ds = \nu^{-2} du$ we have that
    \[
     \int_0^T \sigma_s^2 ds = \frac{\sigma_0^2}{\nu^2} \int_0^{\nu^2 T} \exp\left( 2\nu W_{u/\nu^2} - u\right)du = \int_0^{\nu^2 T} \exp\left( 2(\nu W_{u/\nu^2} - \frac{u}{2})\right)du.
    \]
    The proof of the result concludes by selecting $\tilde{W}_t = \nu^2 W_{t/\nu^2}$, $\mu = -1/2$, and $\tau = \nu^2 T$.
\end{proof}
This last proposition allows us to write
\[
E[|\xi_T|^{\beta}v_T^{-2\alpha}] = \frac{\sigma_0^{\beta}}{\nu^{\beta}} \left( \frac{\nu^2 T}{\sigma_0}\right)^{\alpha} \mathcal{E}_{\beta, \alpha}, \quad \mathcal{E}_{\beta,\alpha} = E\left[ |e^{W_{\tau} + \mu \tau}-1|^{\beta} (A_{\tau}^{(\mu)})^{-\alpha}\right].
\]
We aim to find a bound for $\mathcal{E}_{\beta, \alpha}$.
\begin{prop}
   There exist $C_{\tau}$ and $J_{\beta, \alpha}$ such that
    \[
    \mathcal{E}_{\beta, \alpha} = C_{\tau} \Gamma(\alpha +1) J_{\beta, \alpha},
    \]
    where
    \[
    C_{\tau} = \frac{e^{-\frac{\tau}{8}} e^{\frac{\pi^2}{2\tau}}}{\sqrt{2\pi^3 \tau}}
    \]
    and
    \[
    J_{\beta, \alpha} = \int_{\R} \int_0^{\infty} \frac{|e^x-1|^{\beta}\exp(-(\alpha+\frac{1}{2})x - \frac{\eta^2}{2\tau})\sinh \eta \sin \frac{\pi \eta}{\tau}}{(\cosh x + \cosh \eta)^{\alpha+1}} d\eta d x.
    \]
\end{prop}
\begin{proof}
    It is clear that
    \[
    \mathcal{E}_{\beta, \alpha} = e^{-\tau/8}\int_{\R} \int_0^{\infty}|e^x-1|^{\beta}e^{-x/2}u^{-\alpha-1}\exp \left( -\frac{1+e^{2x}}{2u}\right)\theta\left( \frac{e^x}{u},\tau\right) du dx.
    \]
    Now we do the change of variables $r = \frac{e^x}{u}$, with $du = -\frac{e^x}{r^2}dr$, so that
    \[
    u^{-\alpha-1} = e^{-(\alpha+1)x}r^{p+1}, \quad \frac{1+e^{2x}}{2u} = r\cosh x, \quad \theta\left( \frac{e^x}{u}, \tau\right) = \theta(r,\tau).
    \]
    Therefore,
    \[
    \mathcal{E}_{\beta, \alpha}=e^{-\tau/8}\int_{\R} |e^{x}-1|^{\beta}e^{-(\alpha +\frac{1}{2})x}\left(\int_0^{\infty} r^{\alpha-1}e^{-r\cosh x}\theta(r,\tau)dr \right)dx.
    \]
    Using the fact that 
    \[
    \theta(r,\tau) = \frac{r}{\sqrt{2\pi^3\tau}} e^{\frac{\pi^2}{2\tau}}\int_0^{\infty}e^{-\frac{\eta^2}{2\tau}}e^{-r\cosh \eta}\sinh \eta \sin \frac{\pi \eta}{\tau} d \eta 
    \]
    we can write, using Fubini's theorem,
    \[
     \mathcal{E}_{\beta, \alpha} = \frac{e^{-\frac{\tau}{8}}e^{\frac{\pi^2}{2\tau}}}{\sqrt{2\pi^3 \tau}}\int_{\R} \int_0^{\infty}\int_0^{\infty} r^{\alpha}e^{-r(\cosh \eta + \cosh x)}|e^x - 1|^{\beta} \exp(-(\alpha + \frac{1}{2})x- \frac{\eta^2}{2\tau}) \sinh \eta \sin \frac{\pi \eta}{\tau} dr d \eta d x.
    \]
    The result now follows from the fact that
    \[
    \int_0^{\infty} r^{\alpha} e^{-r(\cosh x + \cosh \eta)} dr = \frac{\Gamma(\alpha+1)}{(\cosh x + \cosh \eta)^{\alpha + 1}}.
    \]
\end{proof}
The last step towards the construction of the bound relies on the following result.
\begin{prop}
    Let
    \[
    K_{\tau} = \int_0^{\infty} e^{-\frac{\eta^2}{2\tau}}\sinh \eta d \eta.
    \]
    Then, the following two properties hold:
    \begin{itemize}
        \item[(i)] $K_{\tau}$ is well defined.
        \item[(ii)] The following bound for $J_{\beta, \alpha}$ holds
        \[
        |J_{\beta,\alpha}| \leq 2^{\alpha+1}\left(2 + \frac{1}{2\alpha + \frac{3}{2}-\beta} \right) K_{\tau}.
        \]
    \end{itemize}
\end{prop}
\begin{proof}
    In order to prove (i) we will use the fact that $\sinh \eta \leq \frac{1}{2}e^{\eta}$ to conclude that
    \[
    K_t \leq \frac{1}{2}\int_0^{\infty} e^{-\frac{\eta^2}{2\tau} + \eta} = \sqrt{2\pi\tau}e^{\tau/2} < \infty.
    \]
    In order to prove (ii) we first use the bound $|\sin(\theta)|\leq 1$ to obtain
    \[
    |J_{\beta,\alpha}|\leq \int_{\R} \int_0^{\infty} \frac{|e^x -1|^{\beta}\exp(-(\alpha+\frac{1}{2})x- \frac{\eta^2}{2\tau})\sinh \eta}{(\cosh x + \cosh \eta)^{\alpha+1}}d\eta d x.
    \]
    Using the bound
    \[
    \cosh x + \cosh \eta \geq \cosh x \geq \frac{1}{2}e^{|x|}
    \]
    we find that
    \[
    |J_{\beta,\alpha}| \leq 2^{\alpha+1}K_{\tau}\int_{\R} |e^x-1|^{\beta}\exp(-(\alpha + \frac{1}{2})x - (\alpha+1)|x|) dx.
    \]
    We now study this integral in the regimes $x \geq 0$ and $x < 0$.
    \begin{itemize}
        \item In the regime $x \geq 0:$ we have $|e^x -1|^{\beta} = (e^ x -1)^{\beta} \leq e^{\beta x}$, so
        \[
        \int_0^{\infty} |e^x-1|^{\beta}\exp(-(\alpha + \frac{1}{2})x - (\alpha+1)|x|) dx \leq \int_0^{\infty} e^{-(2\alpha + \frac{3}{2} -\beta)x}dx = \frac{1}{2\alpha + \frac{3}{2}-\beta}.
        \]
        \item In the regime $x < 0$: we set $y = -x$ and therefore $|e^x -1|^{\beta} = (1-e^{-y})^{\beta} \leq 1$ so the change of variables $y = -x$ gives
        \[
        \int_{-\infty}^0 |e^x-1|^{\beta}\exp(-(\alpha + \frac{1}{2})x - (\alpha+1)|x|) dx \leq \int_0^{\infty} e^{-y/2} dy = 2.
        \]
    \end{itemize}
    Adding the integrals over $x \geq 0$ and over $x < 0$ we obtain the result.
\end{proof}
We now have all the ingredients to prove the convergence in the $\rho$ scale.
\begin{teo}
    The sequence $V^{M_{max}} = \sum_{m=0}^{M_{max}} a_m(\rho)\rho^m$ converges to $V$ as $M_{max} \to \infty$ if $|\rho| < 2^{-\frac{1}{2}} \approx 0.70710678118.$
\end{teo}
\begin{proof}
By replacing $\beta = m$, $2\alpha = m-1$ we find that the bound we obtain for the mixed moments is
\[
E\left[|\xi_T|^m v_T^{-(m-1)} \right] \leq \left(\frac{\sigma_0}{\nu} \right)^m \left( \frac{\nu^2T}{\sigma_0^2}\right)^{\frac{m-1}{2}} C_{\tau} K_{\tau} 2^{\frac{m+5}{2}} \Gamma\left( \frac{m+1}{2}\right),
\]
and therefore
\[
|a_m(\rho)\rho^m| \leq \frac{K \sigma_0 C_{\tau} K_{\tau}}{\sqrt{2\pi}\nu}\sqrt{1-\rho^2} \frac{\sqrt{(m-2)!}\Gamma\left( \frac{m+1}{2} \right)2^{\frac{m+5}{2}}}{m!} \left(\frac{|\rho|}{\sqrt{1-\rho^2}} \right)^m
\]
The quotient criterion for convergence of series states that the series is convergent if
\[
\frac{|\rho|^m}{\sqrt{1-\rho^2}} < 1 \iff |\rho| < \frac{1}{\sqrt{2}},
\]
as desired.
\end{proof}
In order to complete the proof of the convergence of the method, not only do we need the convergence in $\rho$ but also the convergence of the series that define each coefficient $a_m(\rho)$. In this direction, we can slightly modify the bounds.
\begin{prop} \label{prop: new bound}
    If $\alpha \geq \beta$, then
    \[
    |E[\xi_T^{\beta} v_T^{-2\alpha}]| \leq 6 C_{\tau} K_{\tau} \left( \frac{\sigma_0}{\nu}\right)^{\beta}\left( \frac{\nu^2 T}{\sigma_0^2}\right)^{\alpha} 2^{\alpha} \Gamma(\alpha+1).
    \]
\end{prop}
\begin{proof}
    We have, if $\alpha \geq \beta$
    \[
    |\mathcal{E}_{\beta, \alpha}| = C_{\tau}\Gamma(\alpha+1) |J_{\beta,\alpha}| \leq 3C_{\tau} \Gamma(\alpha+1) 2^{\alpha+1} K_{\tau} = 6C_{\tau}K_{\tau} \Gamma(\alpha+1)2^{\alpha}.
    \]
\end{proof}
The bound obtained in Proposition \ref{prop: new bound} allows us to deduce the following convergence criterion for each coefficient $a_m(\rho)$.
\begin{teo} \label{convergence coefficients}
    The series that define the coefficients $a_m(\rho)$ obtained in Theorem \ref{coefs en rho} are convergent if 
    \[
    |X_0 - k| < \frac{\sigma_0 \sqrt{1-\rho^2}}{\nu}.
    \]
\end{teo}
\begin{proof}
    We will show the convergence in the case $m \geq 2$ since the convergence for $m = 0$ and $m = 1$ follow the same argument. Moreover, we will only detail the computations for $m  = 2p+2$ since the $2p+3$ case is analogous. Recall that for $m = 2p+2$ we want to write the coefficient $a_{2p+2}(\rho)$ as
    \[
    \frac{1}{\sqrt{2\pi}} \sum_{n=0}^{\infty}\frac{(-1)^{n+p}(2n+2p)!(X_0-k)^{2n}}{(2n)! (n+p)!2^{n+p}(1-\rho^2)^{n+p+1/2}T^{n+p+1/2}}E\left[ \xi_T^{2p+2} v_T^{-2n-2p-1}\right].
    \]
    We use the bound given in Proposition \ref{prop: new bound}, which applies if $2p+2 \leq n+p+\frac{1}{2}$ (or equivalently, $n \geq p +\frac{3}{2}$) to say that the series defining $a_{2p+2}(\rho)$ converges if the following series converges:
    \[
\frac{1}{\sqrt{2\pi}} \sum_{n=p+2}^{\infty}\frac{(-1)^{n+p}(2n+2p)!(X_0-k)^{2n}}{(2n)! (n+p)!2^{n+p}(1-\rho^2)^{n+p+1/2}T^{n+p+1/2}} \left( \frac{\nu^2T}{\sigma_0^2}\right)^{n+p+\frac{1}{2}} 2^{n+p+\frac{1}{2}} \Gamma(n+p+\frac{3}{2}).
\]
We can make certain useful reductions. The maturities cancel, we can factor out a $(\frac{\nu^2}{\sigma_0^2})^{p+1/2}$, we can also factor out a $(1-\rho^2)^{p+1/2}$ and a $(-1)^{p}$ to reduce the convergence of the previous series to the one of
\[
\sum_{n=p+2}^{\infty} \frac{(2n+2p)!}{(2n)!(n+p)!} \Gamma(n+p+\frac{3}{2}) u^n
\]
where 
\[
u = \frac{-(X_0 - k)^2 \nu^2}{(1-\rho^2) \sigma_0^2}.
\]
To streamline the notation, we set 
\[
a_n = \frac{(2n+2p)!}{(2n)!(n+p)!} \Gamma(n+p+\frac{3}{2})
\]We use the quotient criterion. Let
\[
L = \lim_{n \to \infty} \left|\frac{a_{n+1}}{a_n} \right|.
\]
Then the radius of convergence is $R = 1/L$. Doing basic manipulations with the Gamma function and the factorials it can be easily shown that
\[
a_{n+1} = \frac{(2n+2p+2)(2n+2p+1)(n+p+\frac{3}{2})}{(2n+2)(2n+1)(n+p+1)}a_n.
\]
So $L = R = 1$ and the convergence holds if $|u| < 1$, that is,
\[
|X_0 - k| \leq \frac{\sigma_0\sqrt{1-\rho^2}}{\nu},
\]
as desired.
\end{proof}
\section{Numerical examples}
In this section we will show numerically the performance of our method. As a benchmark for option prices in all models, we will use a Monte Carlo method over the Hull-White formula. That is, our benchmark is
\[
V \approx \sum_{n = 1}^{N_{MC}} Bac(T,X_0 + \rho \xi_T(\omega_n), k, \sqrt{1-\rho^2}v_T(\omega_n))
\]
with $N_{MC} = 400.000$. In order to reduce the variance of the Monte Carlo method, we use antithetic variables on the Brownian motion $W$ that drives $\sigma_t$. In the application of our method, the moments $E[\xi_T^{\beta} v_T^{-2\alpha}]$ are computed via Monte Carlo with a control variate based on the variance swap. That is, we use as a control variate
\[
b\Phi = b\left(\frac{1}{T}\int_0^T \sigma_s^2 - E(\sigma_s^2) ds \right),
\]
where the constant $b$  is chosen so that the variance of
\[
\xi_T^{\beta} v_T^{-2\alpha}-b \Phi
\]
is minimized. Regarding the benchmark method for the computation of the Greeks, we will use finite differences with $N_{MC} = 400.000$ Monte Carlo simulations and difference step $h = 2^{-5}$. The choice of $h$ is for the sake of consistency with \cite{alos2021malliavin}. Since the Bachelier function satisfies
\[
Bac(T, \lambda X_0, k, \lambda \sigma) = \lambda Bac(T,X_0, k, \sigma)
\]
we will assume without loss of generality that $X_0 = 100$.
\begin{ex}
    We first consider the SABR model 
    \[
    \sigma_t = \sigma_0 \exp\left( \nu W_t - \frac{\nu^2 t}{2}\right),
    \]
    with $\sigma_0 = 35$, $\nu = 0.7$ and we choose the correlation parameter $\rho = -0.3$. According to Theorem \ref{convergence coefficients}, we can ensure convergence in the range
    \[
    |X_0- k| < \frac{35}{0.7}\sqrt{1-(-0.3)^2} \approx 47.696960.
    \]
    We choose as a maximum order for the $\rho$ scale expansion $M = 4$ and we choose $N = 50$ as for the maximum order of the expansion defining each coefficient. We apply the method for $400$ strikes equally spaced in the interval $k \in [45, 155]$ for the maturity time $T = 2$ years. In Figure \ref{fig:sabr03} we show the accuracy of the method by plotting the option prices as well as the Bachelier implied volatility smile.
    \begin{figure}[H]
        \centering
        \includegraphics[width=\textwidth]{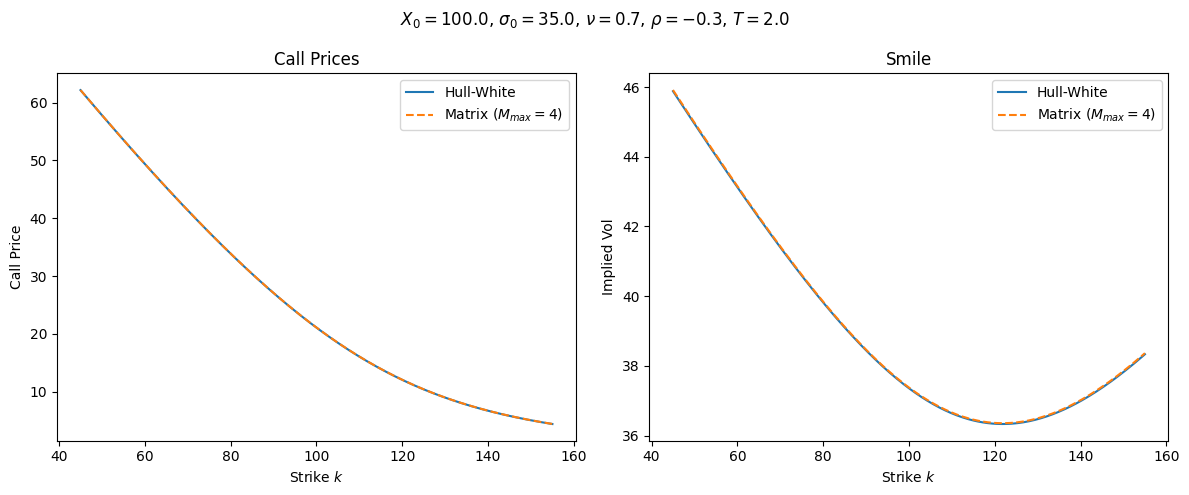}
        \caption{Call prices and implied volatility smile in the SABR model with $M_{max} = 4$ and $N_{max} = 50$.}
        \label{fig:sabr03}
    \end{figure}
    In table \ref{tab:iv_errors_sabr03} we show the relative error between implied volatilities.
\begin{table}[H]
\centering
\caption{Relative error of the implied volatilities in Figure \ref{fig:sabr03} for selected strikes.}
\label{tab:iv_errors_sabr03}
\begin{tabular}{c c}
\toprule
Strike $k$ & Relative error \\
\midrule
45.00  & $7.28\times 10^{-4}$ \\
58.51  & $6.61\times 10^{-4}$ \\
72.29  & $5.83\times 10^{-4}$ \\
86.08  & $5.05\times 10^{-4}$ \\
99.86  & $5.39\times 10^{-4}$ \\
113.65 & $6.44\times 10^{-4}$ \\
127.43 & $6.70\times 10^{-4}$ \\
141.22 & $7.14\times 10^{-4}$ \\
155.00 & $8.31\times 10^{-4}$ \\
\bottomrule
\end{tabular}
\end{table}

    Notice that our method is accurate even when we compute option prices slightly out of the radius of convergence of the coefficients. This does not necessarily confirm that the radius of convergence is greater than the one we stablish, but at least shows the validity of our method as an asymptotic approximation rather than a convergent sequence. The same method allows us to compute the Delta and the Gamma of the options. In Figure \ref{fig:greeks_sabr03} we show how the method is also accurate for the computations of the Greeks.
    \begin{figure}[H]
        \centering
        \includegraphics[width=\textwidth]{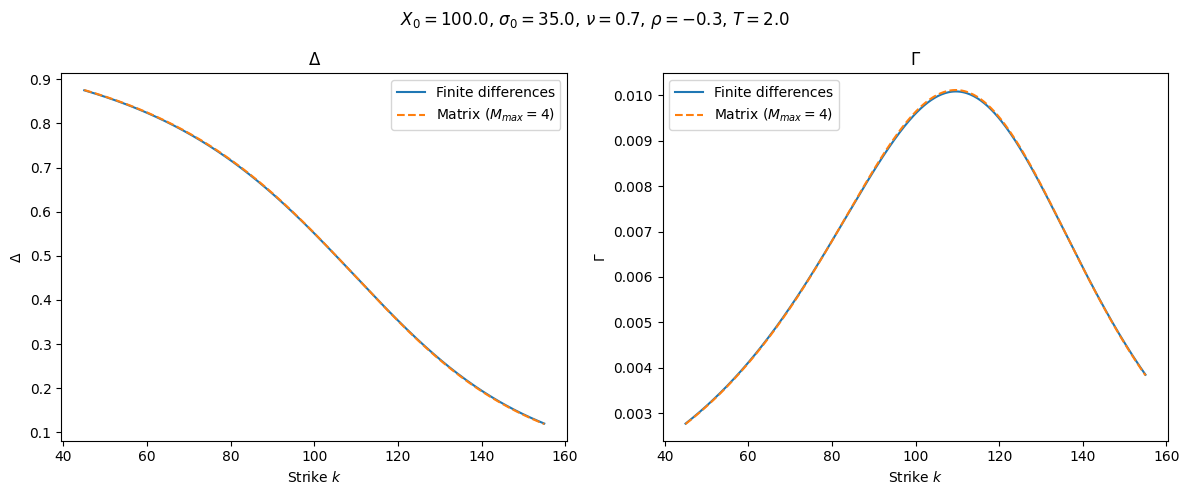}
        \caption{Delta and Gamma of the options with $M_{max} = 4$ and $N_{max} = 50$.}
        \label{fig:greeks_sabr03}
    \end{figure}
\end{ex}
\begin{ex}
    We choose again the SABR model, but this time we choose as parameters $\sigma_0 = 25$, $\nu = 0.4$, $\rho = -0.5$ and 500 options with strikes $k \in [40, 160]$ and time to maturity $T = 5$ years. We apply our method with $M_{max} = 10$, $N_{max} = 50$. In Figures \ref{fig:sabr05} and \ref{fig:greeks_sabr05} we plot the result of the computation of the option prices and the Greeks. 
    
    \begin{figure}[H]
        \centering
        \includegraphics[width=\textwidth]{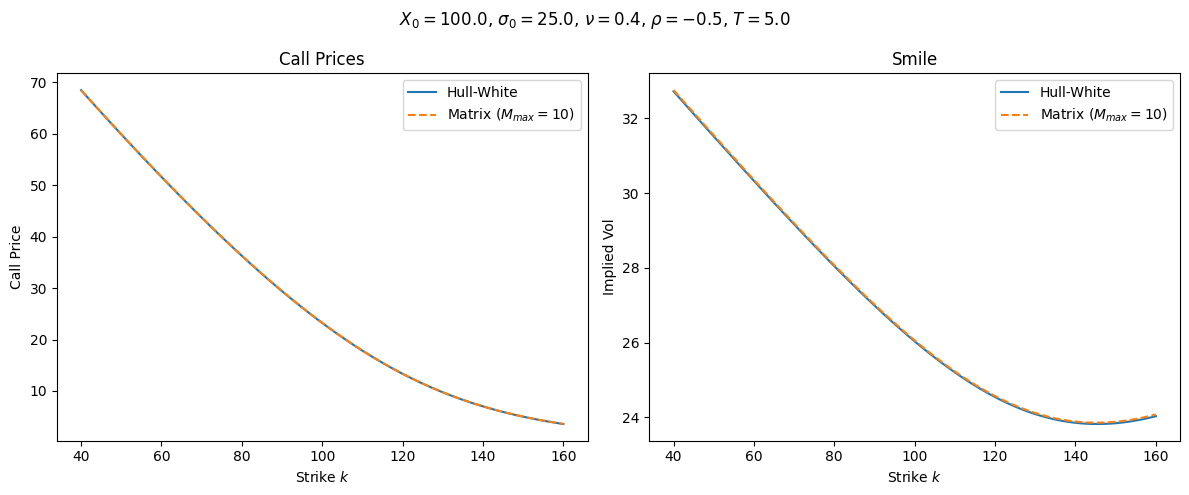}
        \caption{Call prices and implied volatility smile in the SABR model with $M_{max} = 10$ and $N_{max} = 50$.}
        \label{fig:sabr05}
    \end{figure}
    
    In Table \ref{tab:iv_errors_sabr_05} we show the relative errors between the implied volatilities to check that the method is again accurate.
    
\begin{table}[H]
\centering
\caption{Relative errors between implied volatilities of Figure \ref{fig:sabr05}.}
\label{tab:iv_errors_sabr_05}
\begin{tabular}{c c}
\toprule
Strike $k$ & Relative error \\
\midrule
40.00  & $1.02\times 10^{-3}$ \\
54.91  & $9.68\times 10^{-4}$ \\
69.82  & $9.30\times 10^{-4}$ \\
84.97  & $9.22\times 10^{-4}$ \\
99.88  & $9.60\times 10^{-4}$ \\
114.79 & $1.04\times 10^{-3}$ \\
129.94 & $1.22\times 10^{-3}$ \\
144.85 & $1.50\times 10^{-3}$ \\
160.00 & $1.95\times 10^{-3}$ \\
\bottomrule
\end{tabular}
\end{table}
    
    \begin{figure}[H]
        \centering
        \includegraphics[width=\textwidth]{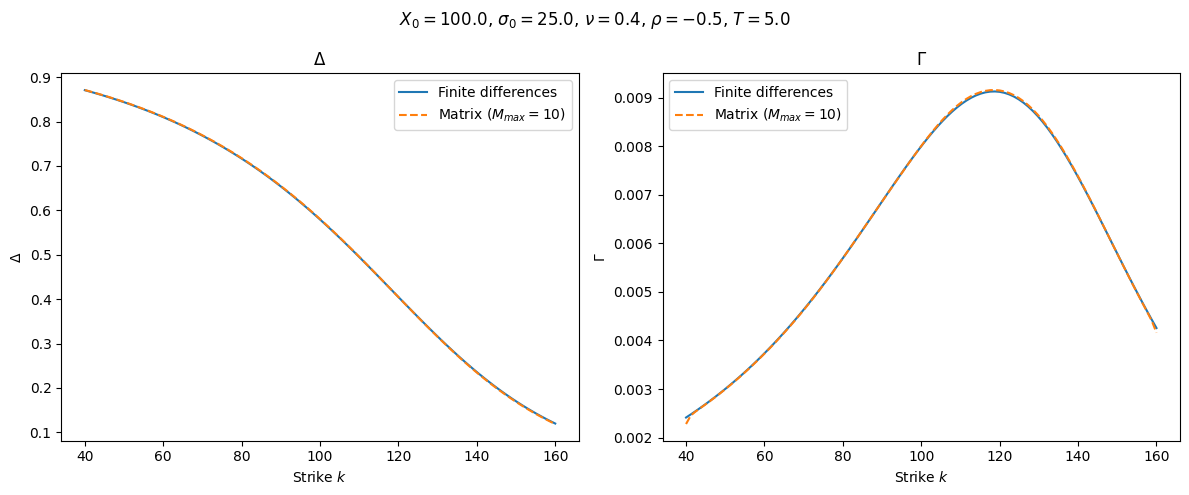}
        \caption{Delta and Gamma of the options with $M_{max} = 10$ and $N_{max} = 50$.}
        \label{fig:greeks_sabr05}
    \end{figure}
\end{ex}

\begin{ex}
    We now consider the rough Bergomi model, that is,
    \[
    \sigma_t^2 = \sigma_0^2 \exp\left( \eta W^H_t - \frac{\eta^2t^{2H}}{2}\right).
    \]
    In this case we haven't derived a result regarding the region of convergence of our method, meaning that for the rough Bergomi model we do not know if our method provides convergent or asymptotic approximation formulas. However, we can still show the validity of our method in this case as long as the correlation parameter $\rho$ stays moderate. We choose as parameters $\sigma_0 = 45$, $\eta = 1.5$, $H = 0.07$. We will test the method with option prices with strikes $k \in [60, 140]$ and time to maturity $T = 0.5$ years. For our method, we choose $M_{max} = 6$ and $N_{max} =50$. In Figures \ref{fig:price_rb04} and \ref{fig:greeks_rb04} we can see the performance of our method.
    \begin{figure}[H]
        \centering
        \includegraphics[width=\textwidth]{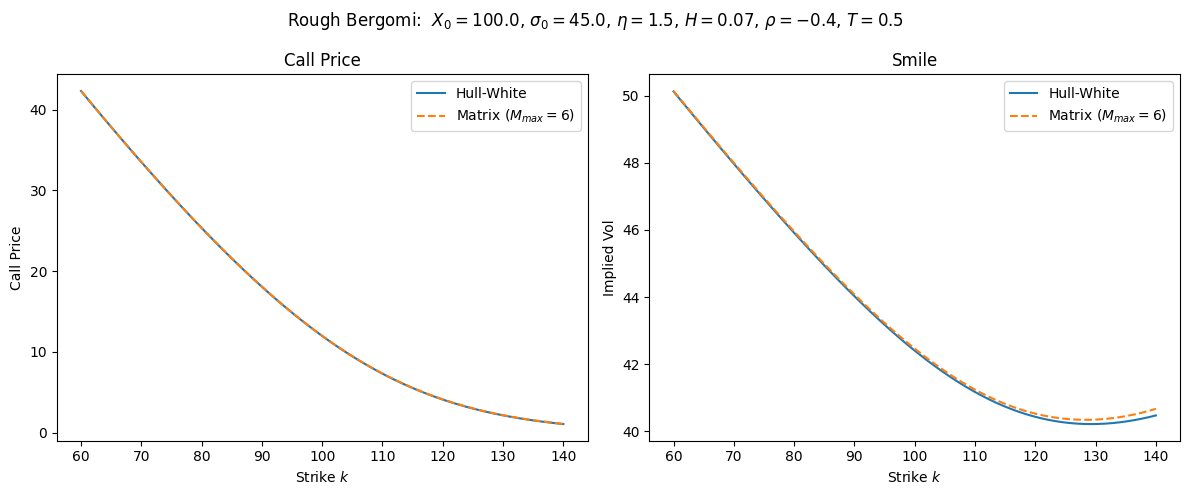}
        \caption{Call prices and implied volatility smile in the rough Bergomi model with $M_{max} = 6$ and $N_{max} = 50$.}
        \label{fig:price_rb04}
    \end{figure}

    In Table \ref{tab:iv_errors_rbergomi} we see the relative errors between the implied volatilities for this example.

    \begin{table}[H]
\centering
\caption{Relative error of the implied volatility in Figure \ref{fig:price_rb04}.}
\label{tab:iv_errors_rbergomi}
\begin{tabular}{c c}
\toprule
Strike $k$ & Relative error \\
\midrule
60.00  & $2.81\times 10^{-4}$ \\
69.82  & $5.71\times 10^{-4}$ \\
79.85  & $8.44\times 10^{-4}$ \\
89.87  & $1.14\times 10^{-3}$ \\
99.90  & $1.41\times 10^{-3}$ \\
109.92 & $1.72\times 10^{-3}$ \\
119.95 & $2.32\times 10^{-3}$ \\
129.97 & $3.28\times 10^{-3}$ \\
140.00 & $4.88\times 10^{-3}$ \\
\bottomrule
\end{tabular}
\end{table}

    \begin{figure}[H]
        \centering
        \includegraphics[width=\textwidth]{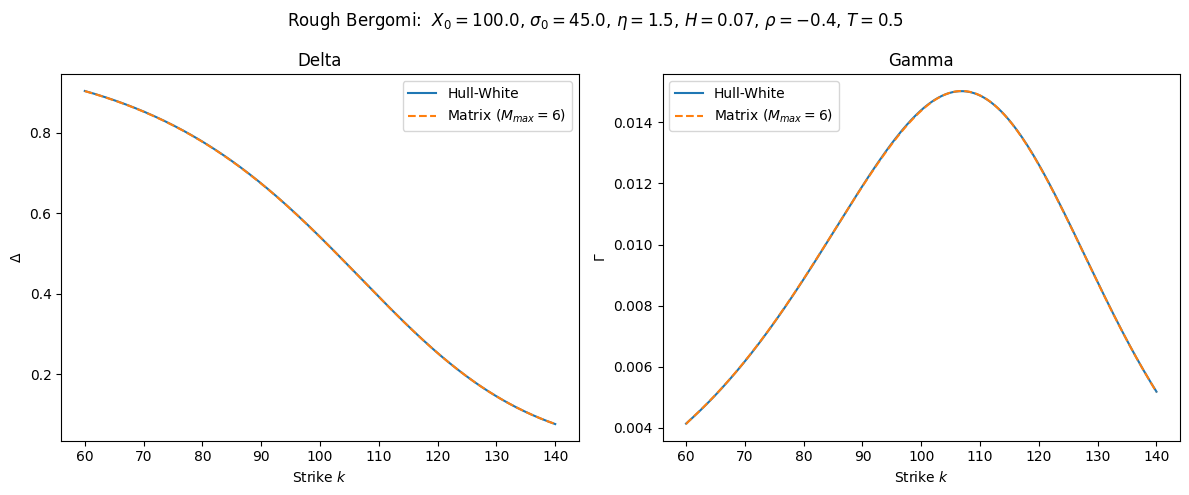}
        \caption{Delta and Gamma of the options with $M_{max} = 6$ and $N_{max} = 50$.}
        \label{fig:greeks_rb04}
    \end{figure}
\end{ex}
\newpage

\section{Conclusions}
In this paper we have applied Taylor expansions and the Itô formula to the Hull-White formula for option pricing in the Bachelier model to derive a numerical method in which, based on elementary linear algebra, we are able to simultaneously compute option prices together with their Delta and Gamma. The number of computations required for the application of the method does not depend on the number of options we want to compute. Therefore, a pre-computation of a finite number of expectations is enough to compute option prices and Greeks for an unlimited amount of strikes within a range of convergence.

\bibliographystyle{apalike}
\bibliography{references.bib}
\end{document}